\documentclass{article}

\usepackage{graphicx}
\usepackage[utf8]{inputenc}
\usepackage[english]{babel}

\usepackage{amsthm}
\usepackage{amsmath}
\usepackage{amssymb}
\usepackage{amsfonts}
\usepackage{hyperref}

\def\R{\mathbb{R}}
\def\N{\mathbb{N}}

\newtheorem{thm}{Theorem}
\newtheorem{lem}{Lemma}
\newtheorem{cor}{Corollary}

\title{Quantum codes do not increase fidelity against isotropic errors}

\author{J. Lacalle\thanks{Dep. de Matem\' atica Aplicada a las TIC, ETS de Ingenier\' ia de Sistemas Inform\' aticos, Universidad Polit\' ecnica de Madrid, C/ Alan Turing s/n, 28031, Madrid, Spain (jesus.glopezdelacalle@upm.es).}, L.M. Pozo-Coronado\thanks{Dep. de Matem\' atica Aplicada a las TIC, ETS de Ingenier\' ia de Sistemas Inform\' aticos, Universidad Polit\' ecnica de Madrid, C/ Alan Turing s/n, 28031, Madrid, Spain (lm.pozo@upm.es).},\\ A.L. Fonseca de Oliveira\thanks{Facultad de Ingenier\' ia, Universidad ORT Uruguay, Montevideo, Uruguay (fonseca@ort.edu.uy).}, R. Mart\' in-Cuevas\thanks{Programa de Doctorado en Ciencias y Tecnolog\' ias de la Computaci\' on para Smart Cities, ETS de Ingenier\' ia de Sistemas Inform\' aticos, Universidad Polit\' ecnica de Madrid, C/ Alan Turing s/n, 28031, Madrid, Spain (r.martin-cuevas@alumnos.upm.es).}}

\date{January 2022}

\begin{document}

\maketitle

\begin{abstract}
Given an $m-$qubit $\Phi_0$ and an $(n,m)-$quantum code $\mathcal{C}$, let $\Phi$ be the $n-$qubit that results from the $\mathcal{C}-$encoding of $\Phi_0$. Suppose that the state $\Phi$ is affected by an isotropic error (decoherence), becoming $\Psi$, and that the corrector circuit of $\mathcal{C}$ is applied to $\Psi$, obtaining the quantum state $\tilde\Phi$. Alternatively, we analyze the effect of the isotropic error without using the quantum code $\mathcal{C}$. In this case the error transforms $\Phi_0$ into $\Psi_0$. Assuming that the correction circuit does not introduce new errors and that it does not increase the execution time, we compare the fidelity of $\Psi$, $\tilde\Phi$ and $\Psi_0$ with the aim of analyzing the power of quantum codes to control isotropic errors. We prove that $F(\Psi_0) \geq F(\tilde\Phi) \geq F(\Psi)$. Therefore the best option to optimize fidelity against isotropic errors is not to use quantum codes.\newline
\end{abstract}

\noindent{\it Keywords}: quantum error correcting codes, isotropic quantum computing errors, quantum computing error fidelity, quantum computing error variance

\section{Introduction}

Currently the biggest obstacle to the development of quantum computing continues
to be control of quantum errors. Since the beginnings of quantum computing
in the 90s of the last century one of the main research objectives was to
solve this stumbling block. To address the problem, two fundamental tools were
developed: quantum error correction codes~\cite{CS,St1,Go1,CRSS1,Go2,CRSS2} in combination with
fault tolerant quantum computing~\cite{Sh,St2,Pr1,Go3,KLZ,Ki,AB}. These studies culminated
in the proof of the quantum threshold theorem, which reads as follows:
a quantum computer with a physical error rate below a certain threshold can,
through application of quantum error correction schemes, suppress the logical
error rate to arbitrarily low levels. However, the proof of this theorem depends
on the discretized treatment of quantum errors, inherited from the construction
of quantum codes.

We believe that the quantum error model used for the proof of the quantum
threshold theorem is not general and that the techniques developed to control
quantum errors do not verify the golden rule of error control: correct all small
errors exactly. For example, in the case of the coding of a qubit by means of
the 5-qubit code~\cite{BDSW,LMPZ} it is argued, using error discretization and that this
code exactly corrects errors in any of the qubits, that the error probability goes
from $p$ to $p^2$, once the correction circuit has been applied. But, what is actually
happening is that the probability of an error (small with high probability) in all
qubits is 1 and that the code cannot correct these simultaneous errors. Then,
an error occurs with probability 1 and, once the correction circuit is applied, it
becomes undetectable.

Therefore it is necessary to make an analysis of quantum errors regardless
of their discretization. The procedure indicated for this is to consider quantum
errors as continuous random variables and characterize them by their corresponding
density functions. In this article we analyze a specific type of error:
isotropic quantum errors. An isotropic error of an $n-$qubit $\Phi$ is one in which
the probability of the error $\Psi$ only depends on the distance between the two
states, $\|\Psi-\Phi\|$, and not on the direction in which the imprecision $\Psi$ occurs
with respect to $\Phi$. They are errors that are easy to analyze due to their central
symmetry with respect to $\Phi$.

In work~\cite{LPF} we have studied the ability of an arbitrary quantum code to correct
these errors, using the variance as the error measure. If $\Phi$ is the $n-$qubit
without error state, $\Psi$ the state resulting from a disturbance modeled by an
isotropic quantum error and $\tilde\Phi$ the result of applying the quantum code correction
circuit, assuming that it does not introduce new errors, the result that we
demonstrate in~\cite{LPF} is the following:
$$
V(\tilde\Phi)\geq V(\Psi),
$$
where $V(\tilde\Phi)=E[\|\tilde\Phi-\Phi\|^2]$ and $V(\Psi)=E[\|\Psi-\Phi\|^2]$ are the variances of the
corrected state $\tilde\Phi$ and the disturbed state $\Psi$ respectively. This means that no
quantum code can handle isotropic errors, or even reduce their variance.

Now we are interested in analyzing the ability of quantum codes to increase
fidelity against isotropic errors, since the fidelity allows to measure the quantum
errors taking into account that the quantum states do not change if they are
multiplied by a phase factor, while the variance used in~\cite{LPF} does not take this
fact into account.

We represent $n-$qubits as points of the unit real sphere of dimension $2d-1$ being $d=2^n$~\cite{NC}, $S^{2d-1}=\{x\in\R^{2d}\ |\ \|x\| =1\}$, taking coordinates with respect to the computational basis $[|0\rangle,|1\rangle,\dots,|2^n-1\rangle]$,
\begin{equation}
\label{For:QubitFormula}
\Psi=(x_0+ix_1,x_2+ix_3,\dots,x_{2d-2}+ix_{2d-1}).
\end{equation}

We consider quantum computing errors as random variables with density
function defined on $S^{2d-1}$. In~\cite{LPF} it is mentioned that it is easy to relate this
representation to the usual representation in quantum computing by density
matrices and that the representation through random variables is more accurate.

We define the variance of a random variable $X$ as the mean of the quadratic
deviation from the mean value $\mu$ of $X$, $V(X)=E[\|X-\mu\|^2]$. In our case, since
the random variable $X$ represents a quantum computing error, the mean value
of $X$ is the $n-$qubit $\Phi$ resulting from an errorless computation. Without loss
of generality, we will assume that the mean value of every quantum computing
error will always be $\Phi = |0\rangle$. To achieve this, it suffices to move $\Phi$ into $|0\rangle$
through a unitary transformation. Therefore, using the pure quantum states
given by Formula (\ref{For:QubitFormula}), the variance of $X$ will be:
\begin{equation}
\label{For:DefVariance}
V(X)=E\left[[\|\Psi-\Phi\|^2\right]=E[2-2x_0]=2-2\int_{S^{2d-1}}x_0f(x)dx.
\end{equation}

Obviously the variance satisfies $V(X)\in[0,4]$. In~\cite{LP} the variance of the sum of two independent errors on $S^{2d-1}$ is presented for the first time. It is proved for isotropic errors and it is conjectured in general that:
\begin{equation}
\label{For:VarFormula}
V(X_1+X_2)=V(X_1)+V(X_2)-\frac{V(X_1)V(X_2)}{2}.
\end{equation}

Considering the representation of errors through random variables, the definition of fidelity is very simple:
\begin{equation}
\label{For:DefFidelity}
F^2(X)=E\left[|\langle\Psi|\Phi\rangle|^2\right]=E\left[x_0^2+x_1^2\right]=\int_{S^{2d-1}}(x_0^2+x_1^2)f(x)dx.
\end{equation}

Then, the problem we want to address is the following: Let $\Phi_0$ be an
$m-$qubit and $\Phi$ the corresponding $n-$qubit encoded by an $(n,m)-$quantum
code $\mathcal C$. Suppose that the coded state $\Phi$ is changed by error, becoming the state $\Psi$.
Now, to fix the error we apply the code correction circuit, obtaining the final
state $\tilde\Phi$. While $\Phi$ is a pure state, $\Psi$ and $\tilde\Phi$ are random variables (mixed states).

We also want to study the possibility of not using quantum codes. In this
case, we suppose that the initial state $\Phi_0$ is changed by error, becoming the
state $\Psi_0$. State $\Psi_0$ is also a random variable. Then our goal is to compare the
fidelities of $\Psi$, $\tilde\Phi$ and $\Psi_0$:
$$
F(\Psi)=E\left[|\langle\Psi|\Phi\rangle|^2\right],\  F(\tilde\Phi)=E\left[|\langle\tilde\Phi|\Phi\rangle|^2\right] \ \text{and}\ F(\Psi_0)=E\left[|\langle\Psi_0|\Phi_0\rangle|^2\right].
$$

In order to compare the fidelities we will assume that the corrector circuit
of $\mathcal C$ does not introduce new errors and it does not increase the execution time.
In other words, we are going to estimate the theoretical capacity of the code to
correct quantum computing errors.

In the case of isotropic errors we shall prove that:
\begin{equation}
\label{For:Result}
F(\Psi_0) \geq F(\tilde\Phi) \geq F(\Psi).
\end{equation}

This result leads us to the conclusion that the best option to optimize fidelity
against isotropic errors is not to use quantum codes. This result goes in the same
direction as that obtained in~\cite{LPF}, which indicates that quantum codes do not
reduce the variance against isotropic errors.

However, the most widely used model of errors in quantum computing is
qubit independent errors. The study of this type of quantum error is much
more complex than that of isotropic errors, because it does not have the same
symmetry. Despite this technical difficulty, we have proved in~\cite{LPFM} that the
$5-$qubit code~\cite{BDSW,LMPZ} is not able to reduce the variance against qubit independent
errors. This result, together with those obtained in~\cite{LPF} and in this article,
clearly reveals the difficulty of the quantum error control challenge and clearly
shows that the continuous nature of quantum errors cannot be ignored.

There are many works related to the control of quantum computing errors, in
addition to those already mentioned above. General studies and surveys on the
subject~\cite{Sc,Pr2,Go4,CDT,HFWH,DP,HZKBR,HKR}, about the quantum computation threshold
theorem~\cite{AGP,WFSH,ACGW,CT}, quantum error correction codes~\cite{OG,LZLX,CCHF,Gu},
concatenated quantum error correction codes~\cite{BPFHC,ES} and articles related to
topological quantum codes~\cite{DPB,NJS}. Lately, quantum computing error control
has focused on both coherent errors~\cite{GM,BEKP} and cross-talk errors~\cite{PSVW,BTS}. Finally,
we cannot forget the hardest error to control in quantum computing, the
quantum decoherence~\cite{Zu}. As we have commented above, these quantum computing
errors can be analyzed in the framework of random variables that has
been set in~\cite{LPF,LP}. In the conclusions we analyze in more detail the characteristics
of the different types of error from the point of view of their control and
in view of the result obtained in this paper.

The outline of the article is as follows: in section 2 we study the fidelity of
the quantum stages $\Psi$, $\Psi_0$ and $\tilde\Phi$; in section 3 we prove the relationship between
them given by Formula (\ref{For:Result}); finally, in section 4 we analyze the conclusions that
can be obtained from the proved result.

\section{Analysis of fidelity}

Associated with the $(n,m)-$quantum code $\mathcal{C}$, the following parameters are defined: $d=2^n$ is the dimension of $\mathcal{C}$, $d^{\,\prime} = 2^m$ and $d^{\,\prime\prime}$ is the number of discrete errors that $\mathcal{C}$ corrects.

First we are going to study how we can compare the fidelity of the quantum states $\Psi$ and $\tilde\Phi$, which are $n-$qubits encoded with the quantum code $\mathcal{C}$, and the fidelity of the state $\Psi_0$, which is an unencoded $m-$qubit state. The working scheme in these two scenarios is illustrated in Figure \ref{Fig:SchemesEncoded-NonEncoded}. We assume that the $\mathcal{C}$ correction circuit, which is applied after each quantum gate in the coded algorithm, does not introduce new errors and is ideally applied for a time $t=0$. In this way we study the theoretical capacity of $\mathcal{C}$ to control isotropic errors, that is, its capacity to increase the fidelity of the final state $\tilde\Phi$ with respect to $\Psi$, and we can compare it with the fidelity of the final state $\Psi_0$ in the scheme without the quantum code  $\mathcal{C}$.

\begin{figure}[h]
\label{Fig:SchemesEncoded-NonEncoded}\quad
\begin{center}
        \includegraphics[scale=0.4]{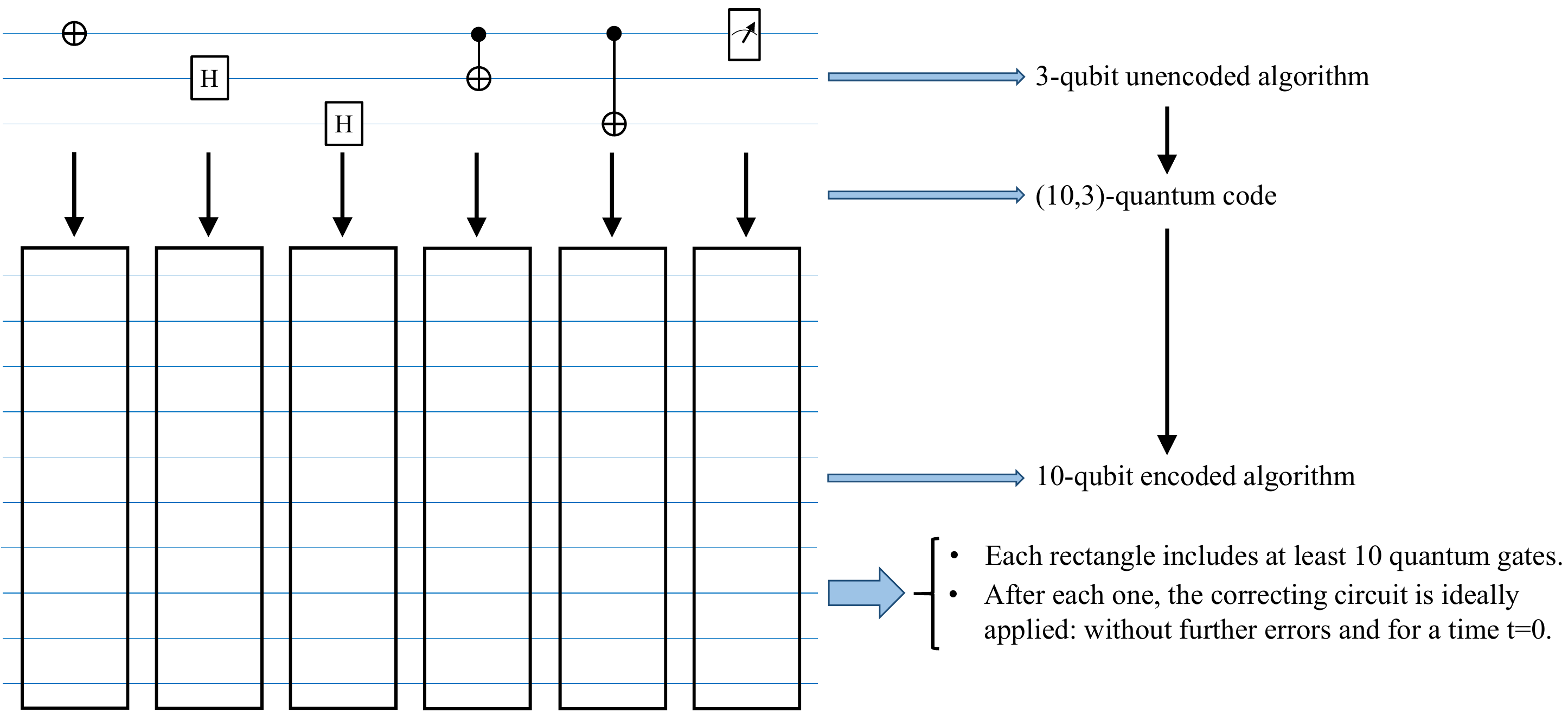}
        \caption{\centerline{Uncoded/coded work scheme.}}
\end{center}
\end{figure}

We analyze the isotropic error as a decoherence error over a unit of time, which corresponds to the time it takes to apply a quantum gate in the coded algorithm. To compare it with the uncoded algorithm we have to bear in mind that the unit of time in this case will be at most the $n-$th part of the unit of time in the coded algorithm. To relate the probability distributions in both cases we use the following equality of variances:
$$
V(E)=V(E_1+E_2+\cdots+E_n),
$$
where $E$ is the decoherence error during a unit of time in the coded algorithm and $E_1$, $E_2$, \dots $E_n$ are independent decoherence errors corresponding to a unit of time in the uncoded algorithm. Using the following generalization of Formula (\ref{For:VarFormula}) demonstrated in~\cite{LP}:
\begin{equation}
\label{For:GeneralVarFormula}
V(E_1+E_2+\cdots+E_n)=2-2\left(1-\frac{v_u}{2}\right)^n,
\end{equation}
where $v_u$ is the variance of each of the independent errors, we obtain the following relation of $v_u$ with the variance $v_c$ of the error $E$:
\begin{equation}
\label{For:RelationshipVariances}
v_c=2-2\left(1-\frac{v_u}{2}\right)^n \quad \Leftrightarrow \quad v_u=2-2\left(\frac{2-v_c}{2}\right)^{1/n}.
\end{equation}

In the case of the normal probability distribution defined in~\cite{LPF,LP}, with the following density function:
\begin{equation}
\label{For:NormalDistribution}
f_n(\sigma,\theta_0)=\frac{(2d-2)!!}{(2\pi)^d}\frac{(1-\sigma^2)}{(1+\sigma^2-2\sigma\cos(\theta_0))^d},
\end{equation}
where the parameter $\sigma$ belongs to the interval $[0,1)$, the above variances have a very simple expression and are independent of the dimension: $v_c=2(1-\sigma_c)$ and $v_u=2(1-\sigma_u)$. The relationship between them given in Formula (\ref{For:RelationshipVariances}) translates into a very simple relationship between the corresponding sigma parameters:
\begin{equation}
\label{For:RelationshipSigmaParameters}
\sigma_c=\sigma_u^n \quad \Rightarrow \quad \sigma_u=\sigma_c^{1/n}.
\end{equation}

From now on we are going to follow the scheme proposed in~\cite{LPF} to calculate the variances of states $\Psi$ and $\tilde\Phi$, but to calculate the fidelities of these states and of state $\Psi_0$.

\subsection{Fidelity of $\Psi$ and $\Psi_0$}

The state $\Psi$, described in Cartesian coordinates in Formula (\ref{For:QubitFormula}) is represented in spherical coordinates as follows:
$$
\begin{array}{l}
\Psi = (\theta_0,\,\theta_1,\,\dots,\,\theta_{2d-2})\quad
\left\{\begin{array}{l}
\vrule height 8pt depth 8pt width 0pt 0\leq\theta_0,\,\dots,\,\theta_{2d-3}\leq\pi \\
\vrule height 12pt depth 2pt width 0pt 0\leq\theta_{2d-2}\leq 2\pi
\end{array}\right., \\
\vrule height 16pt depth 8pt width 0pt x_j=\sin(\theta_0)\,\cdots\,\sin(\theta_{j-1})\,\cos(\theta_j)\quad\text{for all}\quad 0\leq j\leq 2d-2, \\
\vrule height 12pt depth 8pt width 0pt x_{2d-1}=\sin(\theta_0)\,\cdots\,\sin(\theta_{2d-2}).
\end{array}
$$

Using this representation of $\Psi$, the fidelity entered in Formula (\ref{For:DefFidelity}) is as follows:
\begin{equation}
\label{For:DefFidelitySpherical}
F^2(X)=E\left[\cos^2(\theta_0)+\sin^2(\theta_0)\cos^2(\theta_1)\right]=1-E\left[\sin^2(\theta_0)\sin^2(\theta_1)\right].
\end{equation}

\begin{thm}
\label{Thm:FidelityPsi}
The fidelity of the isotropic random variable $\Psi$ with density function $f(\theta_0)$ is equal to:
\begin{equation}
\label{For:FidelityPsi}
F^2(\Psi) = 1 - 4 \, \dfrac{(2\pi)^{d-1}}{(2d-1)!!} \, (d-1) \, {\bar E}\left[\sin^{2d}(\theta_0)\right],
\end{equation}
where $\displaystyle {\bar E}\left[\sin^{2d}(\theta_0)\right]=\int_0^\pi f(\theta_0)\sin^{2d}(\theta_0)d\theta_0$.
\end{thm}
\begin{proof}
We have to calculate the expected value of an expression that depends only on angles $\theta_0$ and $\theta_1$ and the isotropic density function depends only on angle $\theta_0$. Therefore, using Formula (\ref{For:DefFidelitySpherical}):
$$
\begin{array}{lll}
\vrule height 14pt depth 12pt width 0pt F^2(\Psi) & = & \displaystyle 1 - |S^{2d-3}| \, {\bar E}[\sin^{2d}(\theta_0)] \, \int_0^\pi \sin^{2d-1}(\theta_1)d\theta_1 \\
\vrule height 20pt depth 14pt width 0pt           & = & \displaystyle 1 - \dfrac{(2\pi)^{d-1}}{(2d-4)!!} \, 2\dfrac{(2d-2)!!}{(2d-1)!!} \, {\bar E}[\sin^{2d}(\theta_0)] \\
\vrule height 18pt depth 14pt width 0pt           & = & \displaystyle 1 - 4 \, \dfrac{(2\pi)^{d-1}}{(2d-1)!!} \, (d-1) \, {\bar E}[\sin^{2d}(\theta_0)]. \\
\end{array}
$$
We have used equalities from the Appendix.
\end{proof}

\begin{cor}
\label{Cor:FidelityPsiND}
The fidelity of the isotropic random variable $\Psi$ with normal distribution $f_n(\sigma_c,\theta_0)$ is equal to:
\begin{equation}
\label{For:FidelityPsiND}
F^2(\Psi) = \dfrac{1+(d-1)\sigma_c^2}{d}.
\end{equation}
\end{cor}
\begin{proof}
Using the definition of the normal distribution given in Formula (\ref{For:NormalDistribution}) and the Appendix:
$$
\begin{array}{lll}
\vrule height 14pt depth 12pt width 0pt F^2(\Psi) & = & \displaystyle 1 - 4 \, \dfrac{(2\pi)^{d-1}}{(2d-1)!!} \, (d-1) \, {\bar E}[\sin^{2d}(\theta_0)] \\
\vrule height 20pt depth 14pt width 0pt           & = & \displaystyle 1 - 4 \, \dfrac{(2\pi)^{d-1}}{(2d-1)!!} \, (d-1) \, \dfrac{(2d-2)!!}{(2\pi)^d} \, (1-\sigma_c^2) \, \dfrac{(2d-1)!!}{(2d)!!}\pi \\
\vrule height 18pt depth 14pt width 0pt           & = & \displaystyle 1 - \dfrac{d-1}{d} \, (1-\sigma_c^2) = \dfrac{1+(d-1)\sigma_c^2}{d}.  \\
\end{array}
$$
\end{proof}

Theorem \ref{Thm:FidelityPsi} and Corollary \ref{Cor:FidelityPsiND} also apply to state $\Psi_0$, changing the parameter $d$ to $d^\prime$.

\begin{cor}
\label{Cor:FidelityPsi_0}
The fidelity of the isotropic random variable $\Psi_0$ with density function $f(\theta_0)$ is equal to:
\begin{equation}
\label{For:FidelityPsi_0}
F^2(\Psi_0) = 1 - 4 \, \dfrac{(2\pi)^{d^\prime-1}}{(2d^\prime-1)!!} \, (d^\prime-1) \, {\bar E}\left[\sin^{2d^\prime}(\theta_0)\right],
\end{equation}
where $\displaystyle {\bar E}\left[\sin^{2d^\prime}(\theta_0)\right]=\int_0^\pi f(\theta_0)\sin^{2d^\prime}(\theta_0)d\theta_0$. And, if the probability distribution of $\Psi_0$ is normal with density function $f_n(\sigma_u,\theta_0)$, the fidelity is equal to:
\begin{equation}
\label{For:FidelityPsi_0ND}
F^2(\Psi_0) = \dfrac{1+(d^\prime -1)\sigma_u^2}{d^\prime}.
\end{equation}
\end{cor}

To compare the fidelities of $\Psi_0$ and $\tilde\Phi$ we need to obtain their values as a function of their variances $v_u$ and $v_c$ respectively. The relationship between these variances obtained in Formula (\ref{For:RelationshipVariances}) will allow us to relate the fidelities of these states.

\begin{thm}
\label{Thm:LowerBoundFidelityPsi_0}
The fidelity of the isotropic random variable $\Psi_0$ with density function $f(\theta_0)$ satisfy:
\begin{equation}
\label{For:LowerBoundFidelityPsi_0}
F^2(\Psi_0) \geq 1 - \dfrac{2d^\prime-2}{2d^\prime-1}\left(v_u-\left(\dfrac{v_u}{2}\right)^2\right).
\end{equation}
\end{thm}
\begin{proof}
First we prove, similar to the proof of Theorem \ref{Thm:FidelityPsi}, the following:
$$
\begin{array}{lll}
\vrule height 14pt depth 12pt width 0pt F^2(\Psi_0) & = & \displaystyle 1 - |S^{2d^\prime-3}|\, {\bar E}[\sin^{2d^\prime}(\theta_0)] \int_0^\pi\sin^{2d^\prime-1}(\theta_1)d\theta_1 \\
\vrule height 20pt depth 14pt width 0pt             & = & \displaystyle 1 - |S^{2d^\prime-3}|\, {\bar E}[\sin^{2d^\prime}(\theta_0)] \int_0^\pi\sin^{2d^\prime-3}(\theta_1)d\theta_1\, \dfrac{\displaystyle \int_0^\pi\sin^{2d^\prime-1}(\theta_1)d\theta_1}{\displaystyle \int_0^\pi\sin^{2d^\prime-3}(\theta_1)d\theta_1} \\
\vrule height 18pt depth 14pt width 0pt             & = & \displaystyle 1 - \int_{S^{2d^\prime-1}}f(\theta_0)\sin^2(\theta_0)\, \dfrac{\displaystyle \int_0^\pi\sin^{2d^\prime-1}(\theta_1)d\theta_1}{\displaystyle \int_0^\pi\sin^{2d^\prime-3}(\theta_1)d\theta_1} \\
\end{array}
$$

And, using the formulas in the Appendix, we obtain:
$$
F^2(\Psi_0) = \displaystyle 1 - E[\sin^2(\theta_0)]\, \dfrac{2d^\prime-2}{2d^\prime-1}.
$$
\end{proof}

Using Jensen's inequality we obtain a lower bound for $E[\sin^2(\theta_0)]$:
$$
\begin{array}{lll}
\vrule height 14pt depth 12pt width 0pt (E[1-\cos(\theta_0)])^2 & \leq & E[(1-\cos(\theta_0))^2] \\
\vrule height 14pt depth 12pt width 0pt                         & =    & E[1+\cos^2(\theta_0)-2\cos(\theta_0)] \\
\vrule height 14pt depth 12pt width 0pt                         & =    & E[2-2\cos(\theta_0)-\sin^2(\theta_0)] \\
\vrule height 14pt depth 12pt width 0pt                         & =    & v_u - E[\sin^2(\theta_0)]. \\
\end{array}
$$

And then:
$$
E[\sin^2(\theta_0)] \leq v_u - (E[1-\cos(\theta_0)])^2 = v_u-\left(\dfrac{v_u}{2}\right)^2.
$$

Substituting in the formula of $F^2(\Psi_0)$ the previous lower bound of $E[\sin^2(\theta_0)]$, the proof is concluded:
$$
F^2(\Psi_0) \geq 1 - \dfrac{2d^\prime-2}{2d^\prime-1}\left(v_u-\left(\dfrac{v_u}{2}\right)^2\right).
$$

\subsection{Fidelity of $\tilde\Phi$}

The formula for the fidelity of the state $\tilde\Phi$ is very similar to that of the state $\Psi$, Formula (\ref{For:FidelityPsi}), although the proof is more complex because the quantum code $\mathcal{C}$ is involved.

\begin{thm}
\label{Thm:FidelityTildePhi}
The fidelity of the isotropic random variable $\tilde\Phi$ with density function $f(\theta_0)$ is equal to:
\begin{equation}
\label{For:FidelityTildePhi}
F^2(\tilde\Phi) = 1 - 4 \, \dfrac{(2\pi)^{d-1}}{(2d-1)!!} \, (d-d^{\prime\prime}) \, {\bar E}\left[\sin^{2d}(\theta_0)\right],
\end{equation}
where $\displaystyle {\bar E}\left[\sin^{2d}(\theta_0)\right]=\int_0^\pi f(\theta_0)\sin^{2d}(\theta_0)d\theta_0$.
\end{thm}
\begin{proof}
Taking into account Theorem 3 and Corollary 1 of \cite{LPF} the fidelity of $\tilde\Phi$ is the following:
$$
F^2(\tilde\Phi) = E\left[P_0|\langle \Phi| \Pi_0\Psi \rangle|^2\right] + (d^{\prime\prime}-1) \, E\left[P_1|\langle E_1\Phi| \Pi_1\Psi \rangle|^2\right].
$$
where $P_0$ and $P_1$ are the probabilities of measuring the syndromes $0$ and $1$ respectively, $\Pi_0$ and $\Pi_1$ the (normalized) projectors corresponding to the discrete errors $E_0=I$ and $E_1$ associated with the aforementioned syndromes and $E_1\Phi=E_1|0\rangle=|2d^\prime\rangle$.

The first expected value in the above expression is equal to $F^2(\Psi)$ by the Formula (\ref{For:DefFidelitySpherical}) and, using Theorem \ref{Thm:FidelityPsi}, it is obtained:
$$
\begin{array}{lll}
\vrule height 14pt depth 8pt width 0pt E\left[P_0|\langle \Phi| \Pi_0\Psi \rangle|^2\right] & = & E\left[1-\sin^2(\theta_0)\sin^2(\theta_1)\right] \\
\vrule height 16pt depth 12pt width 0pt & = & 1 - 4 \, \dfrac{(2\pi)^{d-1}}{(2d-1)!!} \, (d-1) \, {\bar E}\left[\sin^{2d}(\theta_0)\right]. \\
\end{array}
$$

And the second is the following:
$$
E\left[P_1|\langle E_1\Phi| \Pi_1\Psi \rangle|^2\right] =  E\left[\sin^2(\theta_0) \cdots \sin^2(\theta_{2d^\prime-1})\left(1-\sin^2(\theta_{2d^\prime})\sin^2(\theta_{2d^\prime+1})\right)\right]. \\
$$

Using the Appendix the following is obtained:
$$
\begin{array}{l}
\vrule height 14pt depth 8pt width 0pt \displaystyle E\left[\sin^2(\theta_0)\dots\sin^2(\theta_{2d^\prime-1})\right] = {\bar E}[\sin^{2d}(\theta_0)] \\
\vrule height 20pt depth 14pt width 0pt \displaystyle \cdot \int_0^\pi \sin^{2d-1}(\theta_0) d\theta_1 \cdots \int_0^\pi \sin^{2d-2d^\prime+1}(\theta_{2d^\prime-1}) d\theta_{2d^\prime-1} |S_{2d-2d^\prime-1}| \\
\vrule height 20pt depth 14pt width 0pt \displaystyle = {\bar E}[\sin^{2d}(\theta_0)] \, 2\dfrac{(2d-2)!!}{(2d-1)!!} \, \pi\dfrac{(2d-3)!!}{(2d-2)!!} \cdots 2\dfrac{(2d-2d^\prime)!!}{(2d-2d^\prime+1)!!} \, \dfrac{(2\pi)^{d-d^\prime}}{(2d-2d^\prime-2)!!} \\
\vrule height 20pt depth 14pt width 0pt \displaystyle = {\bar E}[\sin^{2d}(\theta_0)] \, 4 \dfrac{(2\pi)^{d-1}}{(2d-1)!!} \, (d-d^\prime).
\end{array}
$$

Similarly we obtain:
$$
E\left[\sin^2(\theta_0)\dots\sin^2(\theta_{2d^\prime+1})\right] = {\bar E}[\sin^{2d}(\theta_0)] \, 4 \dfrac{(2\pi)^{d-1}}{(2d-1)!!} \, (d-d^\prime-1).
$$

With the last two results the following is obtained:
$$
E\left[P_1|\langle E_1\Phi| \Pi_1\Psi \rangle|^2\right] =  {\bar E}[\sin^{2d}(\theta_0)] \, 4 \dfrac{(2\pi)^{d-1}}{(2d-1)!!}. \\
$$

Finally we get the result we are looking for:
$$
\begin{array}{lll}
\vrule height 14pt depth 8pt width 0pt \displaystyle F^2(\tilde\Phi) & = & \displaystyle E\left[P_0|\langle \Phi| \Pi_0\Psi \rangle|^2\right] + (d^{\prime\prime}-1) E\left[P_1|\langle E_1\Phi| \Pi_1\Psi \rangle|^2\right] \\
\vrule height 20pt depth 14pt width 0pt & = & \displaystyle 1 - {\bar E}[\sin^{2d}(\theta_0)] \, 4 \dfrac{(2\pi)^{d-1}}{(2d-1)!!} \left(d-1 -(d^{\prime\prime}-1) \right) \\
\vrule height 20pt depth 14pt width 0pt & = & \displaystyle 1 - {\bar E}[\sin^{2d}(\theta_0)] \, 4 \dfrac{(2\pi)^{d-1}}{(2d-1)!!} \left(d-d^{\prime\prime}\right). \\
\end{array}
$$
\end{proof}

If the probability distribution of $\Psi$ is normal the fidelity of $\tilde\Phi$ is much simpler.

\begin{cor}
\label{Cor:FidelityTildePsiND}
If $\Psi$ has a normal probability distribution with parameter $\sigma_c$ the fidelity of $\tilde\Phi$ satisfies:
\begin{equation}
\label{For:FidelityTildePsiND}
F^2(\tilde\Phi) = \dfrac{1+(d^\prime -1)\sigma_c^2}{d^\prime}.
\end{equation}
\end{cor}
\begin{proof}
To prove the result, it is enough to substitute in Theorem \ref{Thm:FidelityTildePhi} the value of the integral ${\bar E}[\sin^{2d}(\theta_0)]$ from the Appendix and consider that $d = d^\prime d^{\prime\prime}$.
\end{proof}

To compare the fidelities of $\Psi_0$ and $\tilde\Phi$ we need to obtain $F^2(\tilde\Phi)$ as a
function of the variances $v_c$ of the state $\Psi$.

\begin{thm}
\label{Thm:UpperBoundFidelityTildePsi}
If the state $\Psi$ has an isotropic distribution with density function $f(\theta_0)$ such that:
\begin{equation}
\label{For:Condition_f}
\int_0^\pi (1-\cos(\theta_0))\cos(\theta_0)f(\theta_0)\geq 0,
\end{equation}
the fidelity of $\tilde\Phi$ satisfies:
\begin{equation}
\label{For:UpperBoundFidelityTildePsi}
F^2(\tilde\Phi) \leq 1 - \dfrac{d-d^{\prime\prime}}{2d^\prime-1} \, v_c.
\end{equation}
\end{thm}
\begin{proof}
First we prove, similar to the proofs of Theorems \ref{Thm:FidelityPsi} and \ref{Thm:LowerBoundFidelityPsi_0}, the following:
$$
\begin{array}{lll}
\vrule height 14pt depth 12pt width 0pt F^2(\tilde\Phi) & = & \displaystyle 1 - 4 \, \dfrac{(2\pi)^{d-1}}{(2d-1)!!} \, (d-d^{\prime\prime}) \, {\bar E}\left[\sin^{2d}(\theta_0)\right] \\
\vrule height 20pt depth 14pt width 0pt                 & = & \displaystyle 1 - 4 \, \dfrac{(2\pi)^{d-1}}{(2d-1)!!} \, \dfrac{d-d^{\prime\prime}}{|S^{2d-2}|} \, {\bar E}\left[\sin^{2d}(\theta_0)\right] \, |S^{2d-2}| \\
\vrule height 20pt depth 14pt width 0pt                 & = & \displaystyle 1 - 4 \, \dfrac{(2\pi)^{d-1}}{(2d-1)!!} \, \dfrac{(2d-3)!!}{2(2\pi)^{d-1}} \,  (d-d^{\prime\prime}) \, E\left[\sin^2(\theta_0)\right] \\
\vrule height 20pt depth 14pt width 0pt                 & = & \displaystyle 1 - 2 \, \dfrac{d-d^{\prime\prime}}{2d-1} \, E\left[\sin^2(\theta_0)\right]. \\
\end{array}
$$

Now, using Formula (\ref{For:Condition_f}), we obtain the following lower bound:
$$
\begin{array}{lll}
\vrule height 12pt depth 8pt width 0pt E\left[\sin^2(\theta_0)\right] & =    & \displaystyle E\left[(1-\cos(\theta_0))(1+\cos(\theta_0))\right] \\
\vrule height 18pt depth 12pt width 0pt                               & \geq & \displaystyle E\left[1-\cos(\theta_0)\right] = \dfrac{v_c}{2}. \\
\end{array}
$$

The proof is concluded by introducing the previous lower bound in the expression obtained for $F^2(\tilde\Phi)$.
\end{proof}

\section{Relationship between the fidelity of the states $\Psi_0$, $\tilde\Phi$ and $\Psi$}

The results obtained in the previous section allow us to easily prove the following theorem.

\begin{thm}
\label{Thm:RelationshipTildePhiPsi}
If the state $\Psi$ has an isotropic distribution the following relationship between the fidelities of $\tilde\Phi$ and $\Psi$ holds:
\begin{equation}
\label{For:RelationshipTildePhiPsi}
F^2(\tilde\Phi) \geq F^2(\Psi).
\end{equation}
\end{thm}
\begin{proof}
Theorems \ref{Thm:FidelityPsi} and \ref{Thm:FidelityTildePhi} allow us to prove the result directly, considering that $d-1 \geq d-d^{\prime\prime}$.
\end{proof}

To compare the fidelities of states $\Psi_0$ and $\tilde\Phi$ we need to use Theorems \ref{Thm:LowerBoundFidelityPsi_0} and \ref{Thm:UpperBoundFidelityTildePsi}. However, we must establish a previous result in order to establish the relationship between these states.

\begin{lem}
\label{Lem:InequalityFidelities}
Given $n\in\N$, $n\geq 2$, and $x\in\R$, $0\leq x\leq 4$, the following is satisfied:
$$
g(n,x)=2-2\left(1-\dfrac{x}{2}\right)^n-\left(x-\left(\dfrac{x}{2}\right)^2\right) \geq 0.
$$
\end{lem}
\begin{proof}
The change of variable $\displaystyle y=\left(1-\dfrac{x}{2}\right)$ allows us to better analyze the function:
$$
g(n,y)=1+y^2-2y^n \quad \text{and} \quad x\in[0,4] \ \Leftrightarrow \ y\in[-1,1].
$$

Property $1,\ y^2\geq |y^n|$ for all $y\in[-1,1]$ allows us to conclude that $g(n,y)\geq 0$ for all $y\in[-1,1]$ and this shows that:
$$
g(n,x)\geq 0 \quad \text{for all} \quad x\in[0,4].
$$
\end{proof}

The previous lemma allows us to obtain the main result of this article.

\begin{thm}
\label{Thm:RelationshipPsi_0TildePhi}
If states $\Psi_0$ and $\Psi$ have isotropic distributions with variances $v_u$ and $v_c$ respectively and the density function of $\Psi$ satisfies Formula (\ref{For:Condition_f}), the following relationship between the fidelities of $\Psi_0$ and $\tilde\Phi$ holds:
\begin{equation}
\label{For:RelationshipTildePhiPsi}
F^2(\Psi_0) \geq F^2(\tilde\Phi).
\end{equation}
\end{thm}
\begin{proof}
Theorems \ref{Thm:LowerBoundFidelityPsi_0} and \ref{Thm:UpperBoundFidelityTildePsi} allow us to prove the result, just establishing that the following inequality holds:
$$
\dfrac{d-d^{\prime\prime}}{2d^\prime-1} \, v_c \geq \dfrac{2d^\prime-2}{2d^\prime-1}\left(v_u-\left(\dfrac{v_u}{2}\right)^2\right).
$$

Taking into account that $d=d^\prime d^{\prime\prime}$ the above inequality is equivalent to the following:
$$
v_c \geq \dfrac{2}{d^{\prime\prime}}\left(v_u-\left(\dfrac{v_u}{2}\right)^2\right).
$$

Using the fact that $d^{\prime\prime}\geq 2$ would suffice to prove the first of the following two inequalities:
$$
v_c \geq v_u-\left(\dfrac{v_u}{2}\right)^2 \geq \dfrac{2}{d^{\prime\prime}}\left(v_u-\left(\dfrac{v_u}{2}\right)^2\right).
$$

Substituting the value of $v_c$ given in Formula (\ref{For:RelationshipVariances}) and using the function $g(n,x)$ of Lemma \ref{Lem:InequalityFidelities} we have:
$$
v_c \geq v_u-\left(\dfrac{v_u}{2}\right)^2 \quad \Leftrightarrow \quad g(n,v_u)\geq 0.
$$

Finally, Lemma \ref{Lem:InequalityFidelities} allows us to conclude the proof, using the fact that the variance $v_u\in[0,4]$.
\end{proof}

If the isotropic distributions of $\Psi$ and $\Psi_0$ are normal the condition given in Formula (\ref{For:Condition_f}) for Theorems \ref{Thm:UpperBoundFidelityTildePsi} and \ref{Thm:RelationshipPsi_0TildePhi} is not necessary. Indeed, Corollaries \ref{Cor:FidelityPsiND}, \ref{Cor:FidelityPsi_0} and \ref{Cor:FidelityTildePsiND} clearly imply that:
\begin{equation}
\label{For:GlobalResult}
F(\Psi_0) \geq F(\tilde\Phi) \geq F(\Psi).
\end{equation}

On the other hand, the condition given by Formula (\ref{For:Condition_f}) for Theorems \ref{Thm:UpperBoundFidelityTildePsi} and \ref{Thm:RelationshipPsi_0TildePhi} is a sufficient condition. However, it is not necessary because it has been obtained by underestimating the fidelity of $\Psi_0$ and overestimating that of $\tilde\Phi$. It is verified for very general isotropic distributions, such as for density functions $f(\theta_0)$ that satisfy the following:
$$
f(\theta_0)=0 \quad \text{for all} \quad \theta_0\in\left(\dfrac{\pi}{2},\pi\right].
$$

\begin{figure}[h]
\label{Fig:Curves_F}
\begin{center}
        \includegraphics[scale=0.22]{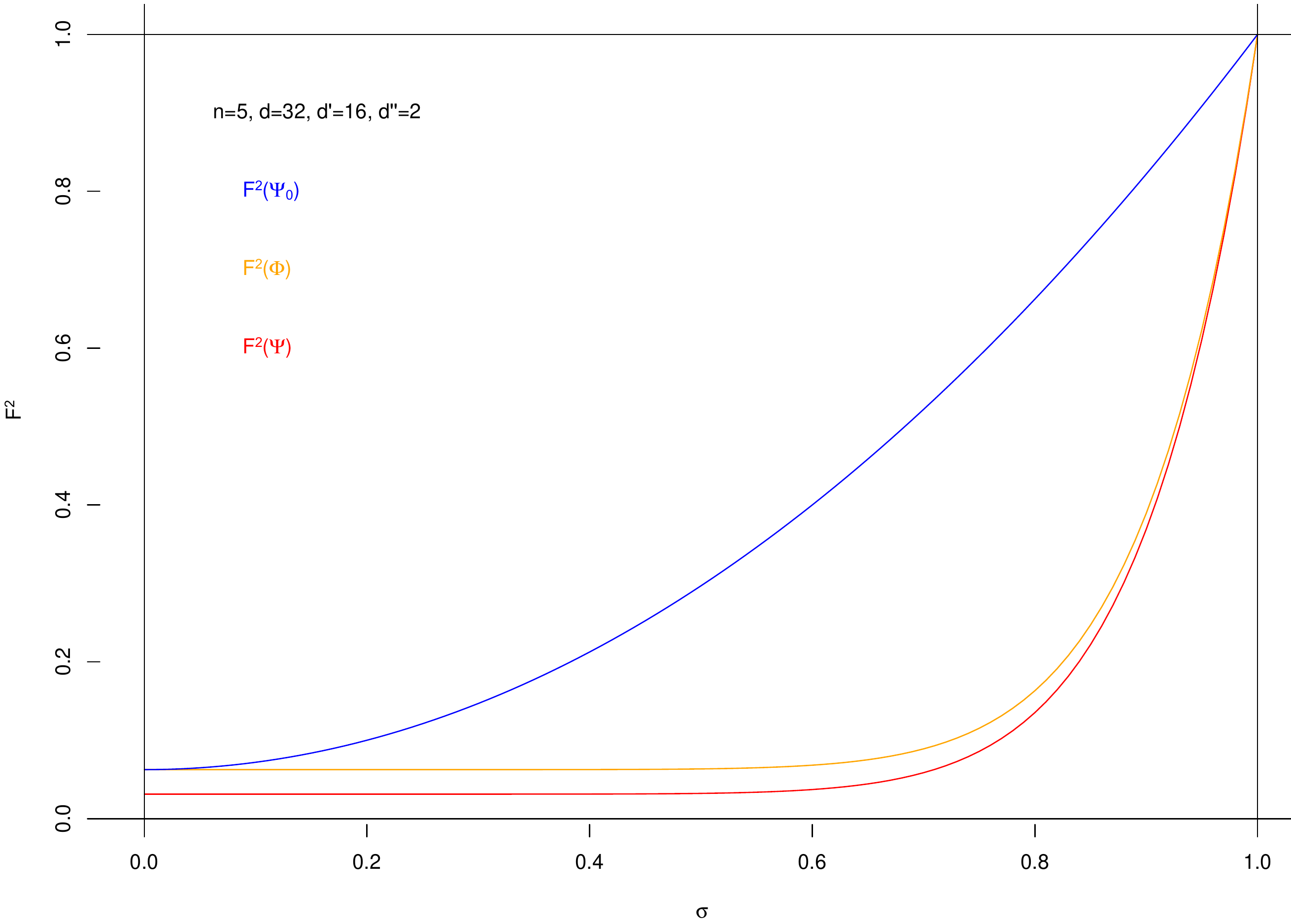} \includegraphics[scale=0.22]{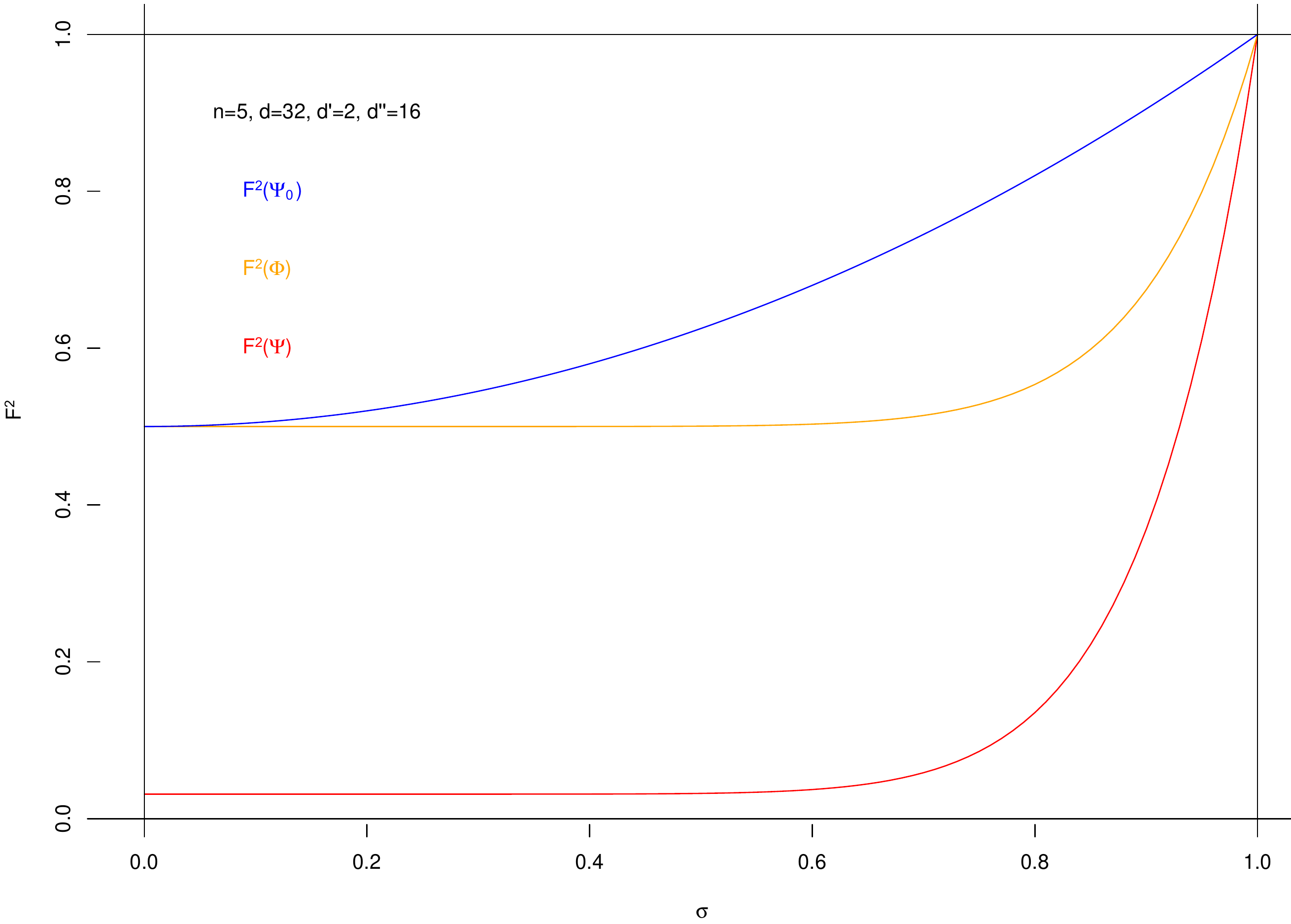}
        \caption{\centerline{Representation of fidelities as a function of $\sigma$.}}
\end{center}
\end{figure}

Figure \ref{Fig:Curves_F} shows the curves of $F^2(\Psi_0)$, $F^2(\tilde\Phi)$ and $F^2(\Psi)$ for normal isotropic distributions and $n=5$ ($d=32$), in the extreme cases $d^\prime=16$ ($d^{\prime\prime}=2$) and $d^\prime=2$ ($d^{\prime\prime}=16$).

The conclusion of the study carried out in this article, in view of the results summarized in Formula (\ref{For:GlobalResult}), is that the best option to obtain the highest fidelity against isotropic errors is not to use quantum codes. On the other hand, the improvement of the fidelity of $\tilde\Phi$ versus that of $\Psi$ seems to be closely related to the dimension of the subspaces to which these states belong: $d^\prime$ for $\tilde\Phi$ versus $d$ for $\Psi$. See Theorems \ref{Thm:FidelityPsi} and \ref{Thm:FidelityTildePhi} and Corollaries \ref{Cor:FidelityPsiND} and \ref{Cor:FidelityTildePsiND}.

\section{Conclusions}

In this article we have analyzed the ability of quantum codes to increase fidelity
of quantum states affected by isotropic decoherence errors. The results obtained,
despite being those expected for this type of quantum errors, are not good from
the point of view of controlling errors in quantum computing. The ability of
quantum codes to reduce errors does not compensate the multiplication of the
number of gates that they require. This fact implies that the best option against
isotropic errors is not to use quantum codes. This result is similar to that
obtained in~\cite{LPF}: quantum codes do not reduce the variance of isotropic errors;
and in~\cite{LPFM}: the $5-$qubit quantum code do not reduce the variance of qubit
independent errors. The last result is more worrying since it negatively affects
the standard model of error in quantum computing. For this reason, it would
be important to study the behavior of fidelity in this case.

These results indicate that continuous errors must be taken into account,
since it is not possible to ensure that the golden rule of error control ``correct all
small errors exactly" is fulfilled. Therefore, the study of the stochastic model
of quantum errors, focused on discrete errors, must be extended to continuous
errors.

For future research, we believe that the continuous quantum computing error
model should be further developed. The results on the ability of quantum codes
to increase the fidelity or to reduce the variance of quantum errors should be
extended to other types of error. It is also important to develop models of the
behavior of quantum errors in highly entangled quantum systems. We need
to know better the behavior of errors in this type of systems so important for
quantum computing. Finally, all these approaches should allow a reformulation
of fault-tolerant quantum computing for continuous errors.

\section{Appendix}

The values of the integrals that have been used throughout the article are included
in this Appendix.

\bigskip

$
\displaystyle\int_{0}^{\pi}\sin^{k}(\theta)d\theta=
\left\{\begin{array}{l}
\displaystyle 2\,\frac{(k-1)!!}{k!!}\quad\mbox{}\quad k=1,\,3,\,5,\,\,\dots \\ \\
\displaystyle \pi\,\frac{(k-1)!!}{k!!}\quad\mbox{}\quad k=2,\,4,\,6,\,\,\dots
\end{array}\right.
$

\bigskip

$
\displaystyle\int_{0}^{\pi}\frac{\sin^{2d-2}(\theta_0)}{(1+\sigma^2-2\sigma\cos(\theta_0))^d}d\theta_0=
\frac{(2d-3)!!}{(2d-2)!!}\frac{\pi}{(1-\sigma^2)}\quad\mbox{}\quad d=1,\,2,\,3,\,\,\dots
$

\bigskip

$
\displaystyle\int_{0}^{\pi}\frac{\cos(\theta_0)\sin^{2d-2}(\theta_0)}{(1+\sigma^2-2\sigma\cos(\theta_0))^d}d\theta_0=
\frac{(2d-3)!!}{(2d-2)!!}\frac{\sigma}{(1-\sigma^2)}\pi\quad\mbox{}\quad d=1,\,2,\,3,\,\,\dots
$

\bigskip

$
\displaystyle\int_{0}^{\pi}\frac{\sin^{2d}(\theta_0)}{(1+\sigma^2-2\sigma\cos(\theta_0))^d}d\theta_0= \frac{(2d-1)!!}{(2d)!!}\pi\quad\mbox{}\quad d=0,\,1,\,2,\,\,\dots
$

\bigskip

Starting from the first integral, the surface of a unit sphere of arbitrary even
$(2d)$ or odd $(2d-1)$ dimension can be calculated.

\bigskip

$
\begin{array}{ccl}
|{\cal S}_{2d}| & = & \displaystyle \int_0^{\pi}\cdots\int_0^{\pi}\int_0^{2\pi}\sin^{2d-1}(\theta_0)\,\cdots\,\sin^{1}(\theta_{2d-2})\ d\theta_0\,\cdots\,d\theta_{2d-2}d\theta_{2d-1} \\ \\
 & = & \displaystyle 2\frac{(2d-2)!!}{(2d-1)!!}\ \frac{(2d-3)!!}{(2d-2)!!}\pi\ 2\frac{(2d-4)!!}{(2d-3)!!}\ \cdots\ \frac{(2-1)!!}{2!!}\pi\ 2\frac{(1-1)!!}{1!!}\ 2\pi \\ \\
 & = & \displaystyle \frac{2(2\pi)^d}{(2d-1)!!} \\
\end{array}
$

\bigskip

$
\begin{array}{ccl}
|{\cal S}_{2d-1}| & = & \displaystyle \int_0^{\pi}\cdots\int_0^{\pi}\int_0^{2\pi}\sin^{2d-2}(\theta_0)\,\cdots\,\sin^{1}(\theta_{2d-3})\ d\theta_0\,\cdots\,d\theta_{2d-3}d\theta_{2d-2} \\ \\
 & = & \displaystyle \frac{(2d-3)!!}{(2d-2)!!}\pi\ 2\frac{(2d-4)!!}{(2d-3)!!}\ \cdots\ \frac{(2-1)!!}{2!!}\pi\ 2\frac{(1-1)!!}{1!!}\ 2\pi \\ \\
 & = & \displaystyle \frac{(2\pi)^d}{(2d-2)!!} \\
\end{array}
$

\end{document}